\documentclass[letterpaper, 10 pt, conference]{ieeeconf}  

\IEEEoverridecommandlockouts                              

\usepackage{amsmath,amsfonts,amssymb}
\usepackage{graphicx}      

\usepackage{float}
\usepackage[usenames]{color}
\usepackage[colorlinks=true, allcolors=blue]{hyperref}
\usepackage{comment}

\usepackage{mathtools}
\usepackage{multirow}
\let\labelindent\relax
\usepackage{enumitem}

\usepackage{balance}

\usepackage{algorithm}

\usepackage{colortbl}

\newcommand\calc{\ensuremath{\mathcal C}}

\newcommand\kset{\ensuremath{\mathbb K}}

\newcommand{\R}{\mathbb{R}}

\newcommand{\N}{\mathbb{N}}

\newcommand{\argmin}{\operatorname{argmin}}

\newcommand{\0}{0}

\newcommand{\1}{{\bf 1}}

\newcommand{\<}{\leqslant}
\renewcommand{\>}{\geqslant}

\newtheorem{lemma}{Lemma}
\newtheorem{theorem}{Theorem}
\newtheorem{proposition}{Proposition}
\newtheorem{remark}{Remark}
\newtheorem{corollary}{Corollary}
\newtheorem{definition}{Definition}

\newtheorem{assumption}{Assumption}

\usepackage[pass]{geometry} 
\title{\LARGE \bf
Robust stabilization of time-delay discrete switched affine systems via a predictive switching control law*
}

\author{Gerson Portilla$^{1}$, Carolina Albea$^{1}$ and Alexandre Seuret$^{1}$
\thanks{*The work of all authors was supported by grants PID2023-150118OB-I00 and ATR2023-145067 funded by MICIU/ AEI /10.13039/501100011033/.}
\thanks{$^{1}$The authors are with Departamento de Ingeniería de Sistemas y Automática, Universidad de Sevilla, Camino de los Descubrimientos, 41092, Seville, Spain
        {\tt\small gportilla,albea,aseuret@us.es}}%
}

\begin{document}
\maketitle
\thispagestyle{empty}
\pagestyle{empty}
\begin{abstract}
This paper addresses the robust control of uncertain discrete-time switched affine systems subject to a single unitary input delay. The unique feature of this class of systems lies in the fact that the control input is the switching signal, which belongs to a finite set of values. Consequently, the closed-loop trajectories do not converge to an equilibrium point but rather to a limit cycle. To mitigate the impact of the input delay, we propose a min-switching predictive control approach, which is based on the known nominal dynamical characteristics of each mode. The objective of this approach is to ensure the robust stabilization of the uncertain system using this nominal predictor, {employing a Lyapunov argument}. Our main result provides { tractable} robust stabilization conditions that guarantee the convergence to a robust limit cycle under system uncertainties and delayed switching. Additionally, an optimization procedure has been incorporated to minimize the size of the attractor, which represents the region where the trajectories asymptotically converge. A numerical example validates the effectiveness of the proposed approach.
\end{abstract}

\section{Introduction}
Switched affine systems constitute a class of hybrid dynamical models that can represent processes subject to abrupt structural changes \cite{sun2011stability}. They are potentially exploited in power electronics \cite{theunisse2015robust}. In this context, the switching mechanism is not merely an external scheduling variable but an intrinsic control input, selected to achieve performance or stability objectives. This approach has motivated extensive research on the stability analysis and control of switched affine systems \cite{deaecto2016stability,egidio2020global,serieye2023attractors}.

A notable feature observed in switched affine systems is the emergence of limit cycles, which arise from the presence of affine terms. Unlike switched linear systems, where the asymptotic behavior is influenced by the stability properties of the subsystem matrices and the switching logic, affine terms may induce convergence to several equilibrium points, resulting in the formation of invariant limit cycles \cite{rubensson2000stability}. The existence of these behaviours is sensitive to the switching sequence and the interactions among subsystems \cite{serieye2023attractors}.

A critical challenge in the implementation of switching strategies arises when delays affect the switching signal, degrading system performance or leading to instability \cite{mahmoud2010}. Although delayed controllers have been studied \cite{vu2010stability,mazenc2017state}, the interaction between switching signal, affine dynamics, and input delays remains comparatively less explored.

Considering the robustness of switched affine systems, it plays a central role in the stability analysis since model uncertainties, parametric variations, and unmodeled affine disturbances can significantly amplify the undesirable effects of switching delays as discussed in \cite{portilla2024}, where the implementation of a state predictor compensating for the input delay failed in the presence of system uncertainties. For hybrid or switched linear systems, traditional robust techniques, as those in \cite{scherer2001theory,ebihara2015s}, often rely on common Lyapunov functions \cite{kouhi2011robust}, multiple Lyapunov functions \cite{serieye2020stabilization}, or dwell-time arguments \cite{allerhand2010robust}, but these approaches do not directly account for the structural mismatch induced by an input delay in the switching signal. 

Recently, robust stabilization conditions for uncertain switched affine systems subject to a unitary input delay were presented in \cite{portilla2026a}. In that work, the Lyapunov theory and a predictive min-switching control law, which is synthesized using a prediction scheme based on nominal system parameters, allowed proving exponential convergence towards a robust limit cycle. This robust limit cycle is characterized by ellipsoids centered at the limit cycle generated from the nominal system parameters. While this research marks a significant advancement in the robust stabilization of delayed switched affine systems, its results are still conservative due to the use of a series of inequalities and the definition of auxiliary variables to derive the stabilization conditions.

In this work, we present robust stabilization conditions for delayed uncertain switched affine systems via a $\min$-switching predictive control law that remains effective under model uncertainties and implementation constraints. Following the ideas of \cite{portilla2026a}, a prediction scheme is considered based on nominal system parameters. The consideration of previous prediction is taken into account to construct a Lyapunov function, which demonstrates the convergence of the system state to a robust limit cycle. Unlike the work presented in \cite{portilla2026a}, we propose a new Lyapunov function whose structure avoids the use of a sequence of Young's inequalities when the negativity of its forward increment is proved. In fact, the robust stabilization conditions are refined and structured in the pursuit of obtaining a less conservative system attractor for delayed uncertain switched affine systems.

 \textbf{Notation:} The set of all real $n\times  m$ matrices is denoted by $\mathbb R^{n\times m}$, and the set of symmetric matrices in $\mathbb R^{n\times n}$ is denoted by $\mathbb S^n$. Matrices $I_n$, $0_{n,m}$ ($0_{n}=0_{n,n}$) and $\1_{n,m}$ denote the identity matrix $\mathbb R^{n\times n}$, null matrix $\mathbb R^{n\times m}$ and all-ones matrix $\mathbb R^{n\times m}$, respectively. When no confusion is possible, the subscripts of the matrices that specify the dimensions are omitted from the equations. For matrix $M$ of $\mathbb S^n$, the notation  $M\succ 0$ indicates that $M$ is positive definite. For any matrix $A=A^{\!\top},B,C=C^{\!\top}$ of appropriate dimensions, matrix $\left[\begin{smallmatrix}A&B\\\star & C  \end{smallmatrix} \right]$ denotes the symmetric matrix $\left[\begin{smallmatrix}A&B\\ B^{\!\top}& C  \end{smallmatrix} \right]$. 
$\text{Co}(A^{\ell},B^{\ell})_{\ell\in\mathbb{L}}$ refers to the convex hull generated by matrices $A^{\ell}$ and $B^{\ell}$, $\ell\in \mathbb L\subset\N$. For a matrix $P\succ 0$ and vector $x\in\R^n$, we denote $\|x\|_P=\sqrt{x^{\!\top} P x}$ as the weighted norm. For matrix $M\succ 0$ and vector $y\in\R^n$, we denote the ellipsoid $\mathcal E(M)=\left\{x\in\mathbb R^n,~ x^{\!\top} M x\< 1 \right\}$.

\section{Problem formulation}\label{sec:form_prob}
\subsection{Uncertain switched affine systems with input-delay}
Consider the discrete switched affine system given by  
\begin{equation}\label{eq:model_x}
\begin{array}{lcl}
	x^+&=&(\bar A_{\sigma}+\delta A_{\sigma}) x+\bar B_{\sigma}+\delta B_{\sigma},\\
 \sigma^+&=& u_\sigma,
\end{array}
\end{equation}
where $x\in\mathbb R^n$ is the system state, $\sigma \in \mathbb K:=\{1,2,\ldots,K\}$ is the switching signal, and $u_\sigma$ is the switching control law that selects the active mode in $\mathbb K$ for any time instant $k\in\N$. The initial conditions of the system \eqref{eq:model_x} are given by  $(x_0,\sigma_0)$ in $\R^n\times \mathbb K$. Here, matrices $\bar A_j\in\mathbb R^{n\times n}$ and $\bar B_j\in\mathbb R^{n\times 1}$ for each mode $j \in \mathbb K=\{1,2,\dots,K\}$ define the \textit{known} and \textit{nominal} system dynamics. On the contrary, matrices $\delta A_j\in\mathbb R^{n\times n}$ and $\delta B_j\in\mathbb R^{n\times 1}$ denote the uncertainties in the dynamics which are assumed to be \textit{unknown} and/or \textit{time-varying}, but they admit a \textit{polytopic representation} given by
\begin{equation}\label{def:polytope}
    [\delta A_{j},\delta B_{j}]\in\text{Co}([\delta A^{\ell}_{j},\delta B^{\ell}_{j}])_{\ell\in\mathbb{L}},
\end{equation}
where $\mathbb{L}$ is a bounded subset of $\N$ and $\delta A^{\ell}_{j}$ and $\delta B^{\ell}_{j}$ are known and constant for any $j\in\mathbb K$ and any $\ell\in\mathbb L$. 

The following notation is adopted along the paper: $x^+=x_{k+1}$, $x=x_{k}$, $u_\sigma=u_{\sigma,k}$ and $\sigma^+=\sigma_{k+1}=u_{\sigma,k}$. Considering this notation, notice that the current switching signal $\sigma_{k}=u_{\sigma,k-1}$ can be interpreted as a switching control law subject to a unitary delay. 

The present study addresses the formulation of an appropriate set-valued map $u_{\sigma}$ aimed at achieving robust stabilization of uncertain switched affine systems subject to a unitary input delay. Its convergence is ensured toward a robust limit cycle determined by a nominal limit cycle obtained from nominal system parameters. 

Here, we adopt the same predictive control structure as in our preliminary paper \cite{portilla2026a}. However, the main novelty relies on an alternative solution for the stability analysis. The main novelty is the definition of another Lyapunov function, whose more flexible structure allows us to give a more accurate specification of the robust limit cycle, estimated by a level set of this new Lyapunov function. 

\subsection{Robust limit cycles}
As mentioned in the previous section, our goal is to achieve exponential convergence to a robust limit cycle, whose concept was initially introduced in \cite{serieye2023attractors}, which extends the notion of stability of limit cycle considered in \cite{egidio2020global}. In this section, we provide some key concepts and a clear definition of this robust limit cycle. To do so, we first introduce the set of periodic functions from $\mathbb N$ to $\mathbb K$: 
$$
		\calc\!=\!\left\{
  \nu \!:\! \mathbb {N} \rightarrow \kset,\mbox{ s.t. } \exists N\!\in\N\!\setminus\!\!\{0\},\forall \ell\in \mathbb {N},\ 
  \nu(\ell\!+\!N)\!=\!\nu(\ell)
  \right\}.
  $$
    
For a given periodic function $\nu \!\in\! \mathcal C$, $N_\nu$ and $\mathbb D_\nu\!=\!\{1,2,\dots, N_\nu\}$ denote the minimum period of $\nu$, i.e. the smallest integer such that $\nu(\ell\!+\!N_\nu)= \nu(\ell)$, $\forall\ell \in \mathbb N$, and the domain of $\nu$, respectively. We also introduce the modulo notation $\lfloor i\rfloor_{\nu}=((i-1) \textrm{ mod }  N_\nu)+1$, for any $i\in\mathbb N$, $i\> 1$. 

To be complete and precise, the following definition of a hybrid limit cycle is first introduced.  
\begin{definition}[Hybrid limit cycle]\cite{serieye2023attractors}
\label{def:limitcycle}
A hybrid limit cycle for system~\eqref{eq:model_x}, or limit cycle in short, is a closed and isolated hybrid trajectory $s:\mathbb N\to\mathbb K \times \mathbb R^n$, $k\mapsto s_k=(\sigma_{k},x_k)$, which is a periodic (but not constant) solution of~\eqref{eq:model_x}.  
\end{definition}  

The previous concept refers to the existence of a periodic and isolated trajectory composed of a sequence of vectors of the state space. 

Considering the nominal system matrices $\bar A_j$ and $\bar B_j$, $\forall j\in\mathbb K$, the concept of limit cycle presented in \cite{egidio2020global,serieye2023attractors} indicates the existence and uniqueness of a sequence of vectors $\{\rho_i\}_{i\in\mathbb D_\nu}$ associated with a given $\nu$ in $\mathcal C$, satisfying 
\begin{equation}\label{def:rho_i}
    \rho_{\lfloor  i+1\rfloor_{\nu}}=\bar A_{\nu( i)}\rho_{ i}+\bar B_{\nu(i)},\quad \forall i\in\mathbb D_\nu.
\end{equation}
Indeed, the set of equations \eqref{def:rho_i} only provides partial information on the state vector $x$ for the uncertain system \eqref{eq:model_x}. Indeed, due to the uncertain matrices $\delta A_j$ and $\delta B_j$, $\forall j\in\mathbb K$, the uniqueness of the set of vectors $\{\rho_i\}_{i\in\mathbb D_\nu}$ is not ensured, but rather they provide a region of possible convergence.

Following the ideas of \cite{serieye2023attractors}, the definition of robust limit cycle is established by considering a neighborhood of the nominal limit cycle, characterized by a subset $\mathcal{L}_i\subset\R^n$, $\forall i\in\mathbb D_\nu$, associated with $\nu\in\mathcal{C}$, which is next recalled.  
\begin{definition}[Robust limit cycle]\cite{serieye2023attractors}\label{def:robust_limit}
    System \eqref{eq:model_x} admits a robust limit cycle associated with $\nu\in\mathcal{C}$, if there exist possibly disjoint subsets $\mathcal{L}_i\subset\R^n$, for $i\in\mathbb D_\nu$, such that
    \begin{equation}\label{eq:robust_inclusion}       (\bar A_{\nu(i)}+\delta A_{\nu(i)}^{\ell})\mathcal{L}_i+\bar B_{\nu(i)}+\delta B_{\nu(i)}^{\ell}\subset \mathcal{L}_{\lfloor  i+1\rfloor_{\nu}},\,\forall\ell\in\mathbb L. 
    \end{equation}
\end{definition}

\begin{remark}\label{remark:singleton}
    As proved in \cite{serieye2023attractors}, for the case of $\delta A_{\nu(i)}^{\ell}=\0$ and $\delta B_{\nu(i)}^{\ell}=\0$, the set of inclusions in \eqref{eq:robust_inclusion} renders the set of equations in \eqref{def:rho_i}, suggesting that the subsets $\mathcal{L}_i$ reduce to singletons $\mathcal{L}_i=\{\rho_i\}$, $i\in\mathbb D_\nu$.
\end{remark}
\subsection{Preliminary result on the delay-free case}
In \cite{serieye2023attractors}, the particular case of nominal and delay-free system is addressed, corresponding to $\delta A_{\nu(i)}^{\ell}=\0$, $\delta B_{\nu(i)}^{\ell}=\0$ and $\sigma=u_{\sigma}$ (not $\sigma^+=u_{\sigma}$). Under these conditions, the authors established sufficient criteria for stabilizing a limit cycle. This result is revisited here, as it provides the theoretical basis for deriving stabilization conditions in the presence of delays and uncertainties. 
\begin{theorem}\cite{serieye2023attractors}\label{th:matrix_W}
For a given cycle $\nu$  in $\mathcal C$, assume that there exist matrices $\{P_i\}_{i\in\mathbb D_\nu}$ in $\mathbb S^n$, such that 
	\begin{equation}\label{eq:LMI}
		P_i\succ 0 , \quad \bar A_{\nu(i)}^{\!\top} P_{\lfloor i+1\rfloor_\nu} \bar A_{\nu(i)}\prec P_i,~\forall i \in\mathbb D_\nu.
	\end{equation}
	Then, the following statements hold:
\begin{enumerate}[label=(\alph*)]
           \item Cycle $\nu$ generates a unique nominal limit cycle for system~\eqref{eq:model_x}. More precisely, there exists a unique sequence of vectors $\{\rho_i\}_{i\in\mathbb D_\nu}$ solution to the set of equations \eqref{def:rho_i}.
			\item $\mathcal A_\nu := \bigcup_{i\in\mathbb D_\nu}  \{\rho_i\}$ is globally exponentially stable (or equivalently is an attractor) for nominal and delay-free system \eqref{eq:model_x} with $\delta A_{\nu(i)}^{\ell}=\0$, $\delta B_{\nu(i)}^{\ell}=\0$, for all $\ell\in\mathbb L$, and $\sigma=u_{\sigma}$, and with the control law 
                \begin{equation}\label{eq:control_nominal}
            		u_{\sigma}(x) \!=\! \left\{ \nu\left( \theta \right),\theta\!\in\!\underset{i\in\mathbb D_\nu}	{\argmin}\left(x\!-\!\rho_i\right)^{\!\top}\!\! P_i\left(x\!-\!\rho_i\right) \right\}.
            	\end{equation}
	        \item If  $\rho_i\neq \rho_j$ for all $i\neq j$ in $\mathbb D_\nu$ holds, then the switching signal $\sigma=u_{\sigma}(x)$ resulting from  \eqref{eq:model_x}, \eqref{eq:control_nominal} converges ultimately to a shifted version of cycle~$\nu$.
\end{enumerate}
\end{theorem}
\begin{remark}
    It is worth mentioning that the stabilization analysis of linear periodic systems in \cite{bolzern1988periodic} arrived at LMI \eqref{eq:LMI} for the first time. In the subsequent work of \cite{egidio2020global} on switched affine systems, the same condition was presented to prove the stability of a limit cycle under a time-varying switching signal.
\end{remark}
\begin{remark}\label{ass:monodromy}
	As discussed in \cite{serieye2023attractors} and \cite{bolzern1988periodic}, condition \eqref{eq:LMI} of Theorem~\ref{th:matrix_W} is equivalent to verifying the Schur stability of the monodromy matrix $\Phi_{\nu}$, defined by $\Phi_{\nu}=\Pi_{\iota=1}^{N_{\nu}} \bar A_{\nu(\iota)}$.
\end{remark}

\subsection{Autonomous augmented delay-free model}\label{subsec:delay_free_model}
Considering that the input delay is constant and known and that delayed discrete-time linear systems are considered as finite-dimensional ones, we here consider an augmented delay-free representation of system \eqref{eq:model_x}, which was initially presented in \cite{portilla2024} to stabilize a delayed switched affine system without uncertainties. Specifically, for a given $\nu$ in $\mathcal C$, system \eqref{eq:model_x} can be rewritten as follows
\begin{equation*}\label{eq:model_xtheta}
	x^+=(\bar A_{\nu(\theta_1)}+\delta A_{\nu(\theta_1)})x+\bar B_{\nu(\theta_1)}+\delta B_{\nu(\theta_1)},\ 
 \theta_1^+=  {u \in\mathbb D_\nu},
\end{equation*}
where $u\in\mathbb D_\nu$ is the new control input, such that $u_\sigma=\nu(u)$ belongs to $\mathbb K$. The presence of the input delay and uncertainties requires defining the previously computed control input and states of the system. To do so, we define memory variables described by
\begin{equation}\label{eq:mem_var}
\begin{split}
    x_1^+=x,\quad
    \theta_1^+=u,\quad 
    \theta_2^+=\theta_1,
\end{split}
\end{equation}
which are updated at each instant. 
The system dynamics can thus be resumed as follows
\begin{equation}\label{eq:model_xitheta}
\begin{split}
	\xi^+&\!\!=\!\!\begin{bsmallmatrix}
	    \bar A_{\nu(\theta_1\!)}\!+\!\delta A_{\nu(\theta_1\!)}&0&0&0\\
     I&0&0&0\\
     0&0&0&0\\
     0&0& 1&0
	\end{bsmallmatrix}\xi\!\!+\!\!\begin{bsmallmatrix}
	    \bar B_{\nu(\theta_1\!)}\!+\!\delta B_{\nu(\theta_1\!)}\\
     0\\0\\0
	\end{bsmallmatrix}\!\!\!+\!\!\!\begin{bsmallmatrix}
	    0\\0\\
     u\\0
	\end{bsmallmatrix},
\end{split}
\end{equation}
for all $\xi\in\mathbb H_{\nu}:= \mathbb R^{2n} \!\times\! \mathbb D^2_\nu,$ where $\xi:=\begin{bmatrix}x^{\!\top} & x_1^{\!\top} & \theta_1 & \theta_2\end{bmatrix}^{\!\top}$.

\subsection{Control objectives}
Following the definition of robust limit cycle and Remark~\ref{remark:singleton}, the system uncertainties can be interpreted as a perturbation around a nominal limit cycle obtained from known and constant system parameters $[\bar A_{j}\ \bar B_{j}]_{j\in\mathbb{K}}$. This argumentation motivates the following assumption considered in our main result.
\begin{assumption}\label{ass:nominal_system}
    For a given cycle $\nu$ in $\mathcal{C}$ and nominal matrices $[\bar A_j\ \bar B_j]_{j\in\mathbb K}$, there exist $\{\rho_i\}_{i\in\mathbb D_\nu}$ in $(\mathbb R^n)^{N_\nu}$ and matrices $\{P_i\}_{i\in\mathbb D_\nu}$ in $\mathbb S^n$, satisfying \eqref{def:rho_i} and \eqref{eq:LMI}.
\end{assumption}
In light of Remark~\ref{ass:monodromy}, Assumption~\ref{ass:nominal_system} can also be formulated as having $\Phi_{\nu}$ Schur stable for a given cycle $\nu \in \mathcal C$.

Under Assumption~\ref{ass:nominal_system}, our control objectives are as follows: 
\begin{itemize}
    \item Design a predictive scheme for compensating a unitary input delay, using the nominal matrices $[\bar A_j\ \bar B_j]_{j\in\mathbb K}$. 
    \item Estimate and minimize the attractor of the closed-loop system, expressed as level sets of a Lyapunov function determined by the nominal limit cycle $\{\rho_i\}_{i\in\mathbb D_\nu}$ and matrices $\{P_i\}_{i\in\mathbb D_\nu}$ in $\mathbb S^n$, for a given $\nu \in \mathcal C$. 
 
\end{itemize}

\section{Robust predictive control design}\label{sec:predic_control}
\subsection{Definition of the predictive control law}
Following ideas of our previous study on predictive control for delayed time-invariant systems \cite{portilla2024}, we introduce the following prediction scheme for system \eqref{eq:model_xitheta}:
\begin{equation}\label{def:predictor_nominal}
    \begin{split}
        \chi_{0|0}(\xi)&=x,\\
    \chi_{1|0}(\xi)&=\bar A_{\nu(\theta_1)}\chi_0(\xi) + \bar B_{\nu(\theta_1)},
    \end{split}
\end{equation}
constructed from nominal system parameters, which aims to predict the variable $x^+$ through $\chi_{1|0}(\xi)$. Also, note that the prediction scheme \eqref{def:predictor_nominal} can be implemented because the nominal matrices $[\bar A_{j}\ \bar B_{j}]_{j\in\mathbb{K}}$ are assumed constant and known. We thus propose the following predictive switching control law to compensate for the input delay:
\begin{equation}\label{eq:control}
            u(\xi)\in\underset{i\in\mathbb D_\nu}{\argmin}\Big( \left\Vert \chi_{1|0}\!-\!\rho_{i}\right\Vert_{P_{i}}^2 \Big).
\end{equation} 

This control law is designed to select the most appropriate mode in $\mathbb K$ based on the predicted state $\chi_{1|0}$. To compensate for the uncertainties induced by $\delta A_{j}$, $\delta B_j$, we introduce the prediction vector made at the previous time instant, given by
\begin{equation}\label{def:predictor_nominal_past}
    \begin{split}
        \chi_{0|1}(\xi)&=x_1,\\
    \chi_{1|1}(\xi)&=\bar A_{\nu(\theta_2)}\chi_{0|1}(\xi) + \bar B_{\nu(\theta_2)},
    \end{split}
\end{equation}
along with the previous control law given by
\begin{equation}\label{eq:past_control}
    \theta_1=\underset{i\in\mathbb D_\nu}{\argmin}\Big( \left\Vert \chi_{1|1}\!-\!\rho_{i}\right\Vert_{P_{i}}^2 \Big).
\end{equation}

The notations $\chi_{r|s}(\xi)$, for $r,s=0,1$, have to be understood as the $r$-step prediction ($r=0,1$) using the current vector $x$ ($s=0$), or its delayed version $x_1$ ($s=1$).  
For simplicity, we will use from now on the short-hand notation $\chi_{r|s}(\xi)=\chi_{r|s}$ and $\chi_{r|s}(\xi^+)=\chi_{r|s}^+$, $\forall r,s=0,1$.

Taking into account that past information is required and that, under initial conditions of $\theta_1$, it is not possible to perform a proper analysis, the following proposition is presented to suitably select system trajectories where the initial conditions of system \eqref{eq:model_xitheta} have been eliminated. 
\begin{proposition}
Any trajectories initiated in $\mathbb H_\nu$ enters 
$$\mathbb H_{\nu}^{\star}:=\left\{\xi\in\mathbb H_{\nu}\mbox{ s.t. }\theta_1=\underset{i\in\mathbb D_\nu}{\argmin}\Big( \left\Vert \chi_{1|1}\!-\!\rho_{i}\right\Vert_{P_{i}}^2 \Big)\right\}$$ 
 after $h=1$ time units.
\end{proposition}


Compared to unperturbed solution provided in our previous paper \cite{portilla2024}, the uncertainties induced by matrices $\delta A_{\nu(i)}$ and $\delta B_{\nu(i)}$ imply a mismatched between the $1$-step prediction and the real value. Nevertheless, it is still possible to establish a relationship between the current and previous predictions, as presented in the following result. 
\begin{lemma}\label{lem:pred}
   The predictors \eqref{def:predictor_nominal} and \eqref{def:predictor_nominal_past} verify $\chi_{s|1}^+=\chi_{s|0}$, $s=0,1$, for any trajectories $(\xi^+,\xi)$ of system \eqref{eq:model_xitheta}.
\end{lemma}
\begin{proof}
    Consider any trajectory $(\xi^+,\xi)$ of system \eqref{eq:model_xitheta}. For $s=0$, system \eqref{eq:model_xitheta} guarantees $\chi_{0|1}^+=x_1^+=x=\chi_{0|0}$ due to \eqref{eq:mem_var}. Consequently, for $s=1$, predictor \eqref{def:predictor_nominal_past} yields 
    \begin{equation*}
        \chi_{1|1}^+=\bar A_{\nu(\theta_2^+)}\underbrace{\chi_{0|1}^+}_{=\chi_{0|0}} + \bar B_{\nu(\theta_2^+)}=\underbrace{\bar A_{\nu(\theta_1)}\chi_{0|0} + \bar B_{\nu(\theta_1)}}_{=\chi_{1|0}}.
    \end{equation*}
    which concludes the proof.
\end{proof}

We also state an essential lemma in which the switched affine system \eqref{eq:model_xitheta} is expressed with respect to the distance to the nominal limit cycle $\{\rho_i\}_{i\in\mathbb D_\nu}$.
\begin{lemma}\label{lemma:pred_error}
    Under Assumption~\ref{ass:nominal_system}, the identity 
    \begin{equation}\label{eq:chi_0plus}
    \chi_{0|0}^+-\chi_{1|0} =\begin{bmatrix}
 T_{\theta_1}&\delta_{\theta_{1}}\\
      \end{bmatrix} \zeta,\quad T_{\theta_1}:=\delta A_{\nu(\theta_{1})}\begin{bmatrix}
 I & I
      \end{bmatrix},
    \end{equation}
holds for any trajectory $(\xi^+,\xi)$ of system \eqref{eq:model_xitheta}, where
\begin{equation}\label{eq:vec_zeta}
\zeta=\begin{bmatrix}(\chi_{1|1}\!-\!\rho_{\theta_1})^{\!\top} & (\chi_{0|0}\!-\!\chi_{1|1})^{\!\top} &  1\end{bmatrix}^{\!\top},
\end{equation}
\begin{equation}
    \delta_{\theta_1}=\delta A_{\nu(\theta_1)}\rho_{\theta_1}+ \delta B_{\nu(\theta_1)}, \quad \theta_1\in\mathbb D_\nu. \label{def_deltaj}
\end{equation}
\end{lemma}
\begin{proof}
For any trajectory $\xi\in\mathbb H_{\nu}$, we have 
\begin{align*}
   \chi_{0|0}^+-\chi_{1|0}=\delta A_{\nu(\theta_1)}\chi_{0|0}+\delta B_{\nu(\theta_1)}.\nonumber
\end{align*}

Introducing $\pm\delta A_{\nu(\theta_1)}\rho_{\theta_1}$ in the previous equation, we get
\begin{align*}
   \chi_{0|0}^+\!-\!\chi_{1|0}&\!=\! \delta A_{\nu(\theta_1)}\chi_{0|0}\!+\!\delta  B_{\nu(\theta_1)} \!+\!  \delta A_{\nu(\theta_1)}\rho_{\theta_1} \!-\! \delta A_{\nu(\theta_1)}\rho_{\theta_1}\nonumber\\
  &= \delta A_{\nu(\theta_1)}(\chi_{0|0}\!-\!\rho_{\theta_1})   + \delta_{\theta_1},
\end{align*}
where $\delta_{\theta_1}$ was defined in \eqref{def_deltaj}. The proof is concluded by noting that 
\begin{align}
    \chi_{0|0}\!-\!\rho_{\theta_1}\!\!=\!(\chi_{0|0}\!-\!\chi_{1|1})\!+\!(\chi_{1|1}\!-\!\rho_{\theta_1})\!=\!\begin{bmatrix}
 I_n &\!\!\!\! I_n&\!\!\!\!0_{n,1}\\
      \end{bmatrix}\!\zeta.\nonumber
\end{align}
\end{proof}


A final property of the predictor \eqref{def:predictor_nominal} is presented, exploiting  Theorem~\ref{th:matrix_W} and the existence of $\{\rho_i\}_{i\in\mathbb D_\nu}$ verifying \eqref{def:rho_i}.
\begin{lemma}\label{lemma:chi_rho}
   Under Assumption~\ref{ass:nominal_system}, the identity 
    \begin{equation}\label{eq:chi_rho}
    \chi_{1|0}-\rho_{\lfloor \theta_1+1\rfloor_{\nu}} \!\!=\!\!\begin{bmatrix} M_{\theta_1}&0
\end{bmatrix}\zeta,~M_{\theta_1}= \begin{bmatrix}
 \bar A_{\nu(i)} & \bar A_{\nu(i)}
      \end{bmatrix},
    \end{equation}
holds for any trajectory $(\xi^+,\xi)$ of system \eqref{eq:model_xitheta}.
\end{lemma}
\begin{proof}
Using the definition of the predictor \eqref{def:predictor_nominal} and limit cycle \eqref{def:rho_i}, we get 
\begin{align*}
\chi_{1|0}\!-\!\rho_{\lfloor \theta_1+1\rfloor_{\nu}}
&=\bar A_{\nu(\theta_1)}(\chi_{0|0}-\rho_{\theta_1})\\
&=\bar A_{\nu(\theta_1)}\left((\chi_{0|0}-\chi_{1|1})+(\chi_{1|1}-\rho_{\theta_1})\right).
\end{align*}

Recalling the augmented vector $\zeta$ in \eqref{eq:vec_zeta}, we finally obtain
\begin{align*}
\chi_{1|0}\!-\!\rho_{\lfloor \theta_1+1\rfloor_{\nu}}&\!=\!\! \begin{bmatrix} \bar A_{\nu(\theta_1)}&\bar A_{\nu(\theta_1)}&\0_{n,1}
\end{bmatrix}\zeta \!=\!\!\begin{bmatrix} M_{\theta_1}&\0_{n,1}
\end{bmatrix}\!\zeta.
\end{align*}
\end{proof}

\subsection{Main result}
We are now able to state the main result of this paper.
\begin{theorem}\label{th:theo3}
Under Assumption~\ref{ass:nominal_system}, i.e., there exist $\{\rho_i,P_i\}_{i\in\mathbb D_\nu}$ satisfying \eqref{def:rho_i}, \eqref{eq:LMI} and for a given parameter $\gamma\in(0,1)$, assume that there exist decision variables $\mu>0$ and matrices $Z_1\in\mathbb S^n$, $Z_2\in\mathbb R^{n}$ and $Z_3\in\mathbb S^{n}$, such that
\begin{equation}\label{eq:LMI_V_positive_unit}
    Z_1\succ 0,\quad Z_3\succ 0,\quad \Omega_i=\begin{bmatrix}
 \mu P_{i}-Z_1 & Z_2 \\ 
 \star & Z_3
\end{bmatrix}\succ 0,
\end{equation}
\begin{equation}\label{eq:condition_rede_unit}
    \Psi_{i,\ell}\succ 0,\quad \forall i \in\mathbb D_\nu,\quad \ell\in\mathbb L,
\end{equation}
where 
\begin{equation*}
\Psi_{i,\ell}=\begin{bsmallmatrix}
    \!\!\!(1-\gamma)\Omega_{i} & 0 & \mu M_{i}^{\!\top} P_{\lfloor i+1\rfloor_{\nu}}& \begin{bsmallmatrix}
            \delta A_{\nu(i)}^{\ell\top}Z_2^{\!\top} \\ \delta A_{\nu(i)}^{\ell\top}Z_2^{\!\top}
        \end{bsmallmatrix}  & \begin{bsmallmatrix}
            \delta A_{\nu(i)}^{\ell\top}Z_3 \\ \delta A_{\nu(i)}^{\ell\top}Z_3
        \end{bsmallmatrix}\\
    \star & \gamma & 0 & \delta_i^{\ell\top} Z_2^{\!\top} &\delta_i^{\ell\top}Z_3\\
    \star &\star & \mu P_{\lfloor i+1\rfloor_{\nu}} & 0 &0\\
    \star &\star & \star & Z_1 & 0\\
    \star &\star & \star & \star & Z_3\\
    \end{bsmallmatrix},
\end{equation*}
with $M_i=\begin{bmatrix} \bar A_{\nu(i)} & \bar A_{\nu(i)}
      \end{bmatrix}$, $\delta_{i}=\delta A_{\nu(i)}^{\ell}\rho_{i} +   \delta B_{\nu(i)}^{\ell}$.
Then, 
\begin{equation*}
    \begin{split}
        \mathcal S_\nu \!\!:=\!\! \left\{ \xi\!\in\!\mathbb H_\nu^\star, \mbox{ s.t. } \begin{bsmallmatrix}
    \!\chi_{1|1}\!-\!\rho_{\theta_1}\\ \!\chi_{0|0}\!-\!\chi_{1|0}
\end{bsmallmatrix}^{\!\!\top}\!\begin{bsmallmatrix}
 \!\mu P_{\theta_1}\!-\!Z_1 & Z_2 \\ 
 \star & Z_3
\!\end{bsmallmatrix}\begin{bsmallmatrix}
    \!\chi_{1|1}\!-\!\rho_{\theta_1}\!\\ \!\chi_{0|0}\!-\!\chi_{1|0}
\end{bsmallmatrix}\!\!\<\!1\!\right\}\!\!,
    \end{split}
\end{equation*}
is ultimately robustly globally exponentially stable for system \eqref{eq:model_xitheta} with the predictive switching control law \eqref{eq:control}.
\end{theorem}
\begin{proof}
The proof is based on demonstrating that the set $\mathcal S_\nu$ is an attractor. To do so, we must guarantee that the system state $x~(=\chi_{0|0})$ converges to the past prediction $\chi_{1|1}$, in turn, $\chi_{1|1}$ converges to a neighborhood of the nominal limit cycle $\rho_{\theta_1}$ through the predictive control law \eqref{eq:control}. Moreover, we need to provide the invariance of $\mathcal{S}_{\nu}$.

We propose the candidate Lyapunov function given by
\begin{equation*}
\begin{split}
V(\xi)\!=\begin{bmatrix}
    \chi_{1|1}-\rho_{\theta_1}\\ \chi_{0|0}-\chi_{1|1}
\end{bmatrix}^{\!\top}\begin{bmatrix}
 \mu P_{\theta_1}-Z_1 & Z_2 \\
 \star & Z_3
\end{bmatrix}\begin{bmatrix}
    \chi_{1|1}-\rho_{\theta_1}\\ \chi_{0|0}-\chi_{1|1}
\end{bmatrix},
\end{split}
\end{equation*}
for all $\xi\in\mathbb H_{\nu}$, where vectors $\{\rho_i\}_{i\in\mathbb D_\nu}$ and matrices $\{P_i\}_{i\in\mathbb D_\nu}$ in $\mathbb S^n$ are solution to \eqref{def:rho_i} and \eqref{eq:LMI} (Assumption~\ref{ass:nominal_system}), respectively. Moreover, it is clear that $V(\xi)>0$ due to \eqref{eq:LMI_V_positive_unit}. 

Define the forward increment of the Lyapunov function as $\Delta V(\xi):=\displaystyle V(\xi^+)-V(\xi)$. Let us focus on the first term $V(\xi^+)$. From its definition, we have 
\begin{align*}
	V(\xi^+)
    &=\begin{bmatrix}
    \chi_{1|1}^+-\rho_{\theta_1^+}\\ \chi_{0|0}^+-\chi_{1|1}^+
\end{bmatrix}^{\!\top}\begin{bmatrix}
 \mu P_{\theta_1^+}-Z_1 & Z_2 \\
 \star & Z_3
\end{bmatrix}\begin{bmatrix}
    \chi_{1|1}^+-\rho_{\theta_1^+}\\ \chi_{0|0}^+-\chi_{1|1}^+
\end{bmatrix}.
\end{align*}

Using Lemmas~\ref{lem:pred} and \ref{lemma:pred_error} and considering that $Z_1\succ\0$, the dynamics of system \eqref{eq:model_xitheta} states that  $\theta_1^+=u$, which yields
\begin{align*}
	V(\xi^+)&=\displaystyle \begin{bmatrix}
    \chi_{1|0}-\rho_{u}\\ \begin{bmatrix}
 T_{\theta_1}&\delta_{\theta_{1}}\\
      \end{bmatrix} \zeta
\end{bmatrix}^{\!\top}\begin{bmatrix}
 \mu P_{u}-Z_1 & Z_2 \\
 \star & Z_3
\end{bmatrix}\begin{bmatrix}
    \chi_{1|0}-\rho_{u}\\ \begin{bmatrix}
 T_{\theta_1}&\delta_{\theta_{1}}\\
      \end{bmatrix} \zeta\end{bmatrix}\\
&=\begin{bmatrix}
    \chi_{1|0}-\rho_{u}\\  \begin{bmatrix}T_{\theta_1}&\delta_{\theta_{1}}\\\end{bmatrix}\zeta
\end{bmatrix}^{\!\top}\Bigg(\begin{bmatrix}
 \mu P_{u} & \0 \\ 
 \star & Z_3+Z_2^{\!\top} Z_1^{-1}Z_2
\end{bmatrix}\\
&- \begin{bmatrix}
 Z_1 & -Z_2 \\ 
 \star & Z_2^{\!\top} Z_1^{-1}Z_2
\end{bmatrix}\Bigg)
\begin{bmatrix}
    \chi_{1|0}-\rho_{u}\\  \begin{bmatrix}T_{\theta_1}&\delta_{\theta_{1}}\\\end{bmatrix}\zeta
\end{bmatrix}.
\end{align*}

Then, since $
    \begin{bsmallmatrix}Z_1 & -Z_2 \\\star & Z_2^{\!\top} Z_1^{-1}Z_2\end{bsmallmatrix}\succeq 0$, the following upper bound of $V(\xi^+)$ is derived
leading to the following inequality
\begin{align*}
V(\xi^+)&\!\<\!\begin{bmatrix}
    \chi_{1|0}\!-\!\rho_{u}\\  \begin{bsmallmatrix}T_{\theta_1}&\!\delta_{\theta_{1}}\end{bsmallmatrix}\zeta
\end{bmatrix}^{\!\top}\begin{bmatrix}
 \mu P_{u} & \0 \\ 
 \star & \!\!\!Z_3\!+\!Z_2^{\!\top} \!Z_1^{-1}Z_2
\end{bmatrix}\begin{bmatrix}
    \chi_{1|0}\!-\!\rho_{u}\\  \begin{bsmallmatrix}T_{\theta_1}&\!\delta_{\theta_{1}}\end{bsmallmatrix}\zeta
\end{bmatrix}\\
&=\mu \Vert \chi_{1|0}\!-\!\rho_{u}\Vert_{P_{u}}^2\!\!+\!\zeta^{\!\top} \!\!\begin{bmatrix}T_{\theta_1}^\top \\ 
\delta_{\theta_{1}}^{\!\top}
\end{bmatrix} \! (Z_3\!+\! Z_2^{\!\top} Z_1^{-1}Z_2)\! \begin{bmatrix}T_{\theta_1}^{\!\top} \\ 
\delta_{\theta_{1}}^{\!\top}
\end{bmatrix}^{\!\top}\!\!\!\!\zeta.
\end{align*}

Furthermore, the control law \eqref{eq:control} 
ensures that $\Vert \chi_{1|0}\!-\!\rho_{u}\Vert_{P_{u}}^2 \< \Vert \chi_{1|0}\!-\!\rho_{\iota}\Vert_{P_{\iota}}^2$ holds for any $\iota$ in $\mathbb D_\nu$.  Therefore, selecting $\iota =\lfloor \theta_1+1\rfloor_{\nu}$, we obtain 
\begin{align*}
V(\xi^+)&\<\mu \Vert \chi_{1|0}-\rho_{\lfloor \theta_1+1\rfloor_{\nu}}\Vert_{P_{\lfloor \theta_1+1\rfloor_{\nu}}}^2\\
&\!+\!\zeta^{\!\top} \!\!\begin{bmatrix}T_{\theta_1}^\top \\ 
\delta_{\theta_{1}}^{\!\top}
\end{bmatrix} \! (Z_3\!+\! Z_2^{\!\top} Z_1^{-1}Z_2)\! \begin{bmatrix}T_{\theta_1}^{\!\top} \\ 
\delta_{\theta_{1}}^{\!\top}
\end{bmatrix}^{\!\top}\!\!\!\!\zeta.
\end{align*}

We are now in a position to use Lemma~\ref{lemma:chi_rho} through \eqref{eq:vec_zeta}: 
\begin{align*}
V(\xi^+)&\<\mu \zeta^{\!\top}\begin{bmatrix} M_{\theta_1}&\0_{n,1}
\end{bmatrix}^{\!\top} P_{\lfloor \theta_1+1\rfloor_{\nu}}\begin{bmatrix} M_{\theta_1}&\0_{n,1}
\end{bmatrix}\zeta\\
&\!+\!\zeta^{\!\top} \!\!\begin{bmatrix}T_{\theta_1}^\top \\ 
\delta_{\theta_{1}}^{\!\top}
\end{bmatrix} \! (Z_3\!+\! Z_2^{\!\top} Z_1^{-1}Z_2)\! \begin{bmatrix}T_{\theta_1}^{\!\top} \\ 
\delta_{\theta_{1}}^{\!\top}
\end{bmatrix}^{\!\top}\!\!\!\!\zeta.
\end{align*}

Thus, the forward increment $\Delta V(\xi)$ is bounded as follows
\begin{align*}
	\Delta V(\xi)&\<\displaystyle \mu \zeta^{\!\top}\begin{bmatrix} M_{\theta_1}&\0_{n,1}
\end{bmatrix}^{\!\top} P_{\lfloor \theta_1+1\rfloor_{\nu}}\begin{bmatrix} M_{\theta_1}&\0_{n,1}
\end{bmatrix}\zeta\\
&\!\!\!\!\!\!\!\!\!\!\!\!+\!\zeta^{\!\top} \!\!\begin{bmatrix}T_{\theta_1}^\top \\ 
\delta_{\theta_{1}}^{\!\top}
\end{bmatrix} \! (Z_3\!+\! Z_2^{\!\top} Z_1^{-1}Z_2)\! \begin{bmatrix}T_{\theta_1}^{\!\top} \\ 
\delta_{\theta_{1}}^{\!\top}
\end{bmatrix}^{\!\top}\!\!\!\!\zeta-\zeta^{\!\top}\begin{bsmallmatrix}
 \mu P_{\theta_1}-Z_1 & Z_2 & 0\\
 \star & Z_3 & 0\\
 \star & \star & 0
\end{bsmallmatrix}\zeta.
\end{align*}

As the objective is to ensure the asymptotic stability of $\mathcal S_\nu$, i.e., $V(\xi)\> 1$, we note that inequality $V(\xi)\> 1$ can be rewritten using the augmented vector $\zeta$:
\begin{equation}\label{eq:inv_cond}
    \zeta^{\top}\begin{bmatrix}
        \mu P_{\theta_1}-Z_1 & Z_2 & 0\\
        \star & Z_3 & 0\\
        \star & \star & -1\end{bmatrix}\zeta \>0.
\end{equation}

Therefore, an S-procedure guarantees that condition $\Delta V(\xi) < 0$ for any $\xi\not\in\mathcal S_\nu$ is equivalent to the existence of a positive scalar $\gamma\in(0,1)$ such that
\begin{equation*}
    \Delta V(\xi)+\gamma (V(\xi)-1) \< 0 \quad\quad\Leftrightarrow\quad\quad -\zeta^{\!\top}\Phi_{\theta_1}\zeta\<0,
\end{equation*}
where
\begin{equation*}
\begin{split}
    \Phi_{\theta_1}&=\begin{bmatrix}
 (1\!-\!\gamma)\Omega_{\theta_1} & \!\!\!\!\!\!\0_{n,1}\\
 \star & \!\!\!\!\!\!\gamma
\end{bmatrix}\! -\! \mu \!\begin{bmatrix} M_{\theta_1}&\!\!\!\!\0_{n,1}
\end{bmatrix}^{\!\top}\!\! P_{\lfloor \theta_1\!+\!1\rfloor_{\nu}}\!\begin{bmatrix} M_{\theta_1}&\!\!\!\!\0_{n,1}
\end{bmatrix}\\
&- \!\zeta^{\!\top} \!\!\begin{bmatrix}T_{\theta_1}^\top \\ 
\delta_{\theta_{1}}^{\!\top}
\end{bmatrix} \! (Z_3\!+\! Z_2^{\!\top} Z_1^{-1}Z_2)\! \begin{bmatrix}T_{\theta_1}^{\!\top} \\ 
\delta_{\theta_{1}}^{\!\top}
\end{bmatrix}^{\!\top}\!\!\!\!\zeta
\end{split}
\end{equation*}

Using again the Schur complement, condition $\Phi_{\theta_1}\succ 0$ is also equivalent to $\Psi_{\theta_1}\succ\0$, where
\begin{equation*}
\Psi_{\theta_1}\!\!=\!\!\begin{bsmallmatrix}
    (1-\gamma)\Omega_{\theta_1} & \0 & \mu M_{\theta_1}^{\!\top} P_{\lfloor \theta_1\!+\!1\rfloor_{\nu}}& \begin{bsmallmatrix}
            \delta A_{\nu(\theta_1)}^{\top}Z_2^{\!\top} \\ \delta A_{\nu(\theta_1)}^{\top}Z_2^{\!\top}
        \end{bsmallmatrix} & \begin{bsmallmatrix}
            \delta A_{\nu(\theta_1)}^{\top}Z_3 \\ \delta A_{\nu(\theta_1)}^{\top}Z_3
        \end{bsmallmatrix}\\
    \star & \gamma & \0 & \delta_{\theta_1}^{\top} Z_2^{\!\top} &\delta_{\theta_1}^{\top}Z_3\\
    \star &\star & \mu P_{\lfloor \theta_1\!+\!1\rfloor_{\nu}} & \0 &\0\\
    \star &\star & \star & Z_1 & \0_{n}\\
    \star &\star & \star & \star & Z_3\\
    \end{bsmallmatrix},
\end{equation*}
with $\delta_{\theta_1}$ defined in Lemma~\ref{lemma:pred_error}. The last step consists of recalling that $\delta A_{j}$ and $\delta B_{j}$ admit the polytopic representation \eqref{def:polytope} and noticing that condition $\Psi_{\theta_1}\succ\0$ is linear with respect to $(\delta A_{\nu(\theta_1)},\delta B_{\nu(\theta_1)})$. It suggests that the previous inequality can be rewritten as $\Psi_{\theta_1}=\sum_{\ell\in\mathrm L}\lambda_{\ell}\Psi_{\theta_1,\ell}\succ\0$, for $\lambda_\ell\in(0,1)$ such that $\sum_{\ell\in\mathrm L} \lambda_\ell=1$. Hence, the strict positivity of $\Psi_{\theta_1}\succ\0$ is ensured thanks to conditions \eqref{eq:condition_rede_unit}, proving the robust global exponential convergence to attractor $\mathcal{S}_{\nu}$.

To complete the proof, it remains to show that $\mathcal S_\nu$ is invariant. To do so, notice that if condition \eqref{eq:condition_rede_unit} is satisfied, we already showed that $\Delta V(\xi) + \gamma ( V(\xi)-1)<0$, which implies that 
\begin{equation*}
    \begin{split}
        V(\xi^+)&=V(\xi)-\gamma ( V(\xi)-1)\<(1-\gamma)V(\xi)+\gamma.
    \end{split}
\end{equation*}

Since $\xi\in \mathcal S_\nu$ ($V(\xi)\leq 1$) and $\gamma\in(0,1)$, then $V(\xi^+)\<(1-\gamma)+\gamma=1$ holds true. Hence, $\xi^+\in \mathcal S_\nu$.
\end{proof}
\begin{remark}
    It is worth recalling that in Theorem~\ref{th:matrix_W}, $P_i$ and $\rho_i$ are fixed apriori due to Assumption~\ref{ass:nominal_system}. The only decision variables are $\mu$ (a scaling parameter of $P_i$) and $Z_i$. Also, note that conditions $\Psi_{i,\ell}\succ\0$ are not LMI due to $\gamma\in(0,1)$. However, one can fix it and perform a gridding.
\end{remark}

Since the main motivation of this work is the achievement of a less conservative attractor $\mathcal{S}_{\nu}$ for system \eqref{eq:model_xitheta}, we next derive an optimization problem, leveraging the decision variable $\mu$. 
\begin{corollary}\label{op:problem} Under Assumption~\ref{ass:nominal_system} and for a given parameter $\gamma\in(0,1)$ and $\alpha>0$, assume that there exist decision variables $\mu>0$ and matrices $Z_1\in\mathbb S^n$, $Z_2\in\mathbb R^{n}$ and $Z_3\in\mathbb S^{n}$, solution to the following optimization problem
    \begin{equation*}
\begin{aligned}
\max \ &\alpha \sum_{i\in\mathbb D_\nu}\textup{Tr}(\mu P_i-Z_1)\\
\textrm{s.t.} \ &Z_1\succ \0,\  \Omega_i\succ\0,
\Psi_{i,\ell}\succ \0,\ \forall i \in\mathbb D_\nu,\  \ell\in\mathbb L.
\end{aligned}
\end{equation*}
\end{corollary}
\begin{proof}
    The proof is omitted because it is a direct from proof of Theorem 2. 
\end{proof}
\section{Illustrative example}\label{sec:example}
Consider the switched affine system \eqref{eq:model_x} borrowed from \cite{egidio2019novel}, consisting of two unstable modes given by  
\begin{equation*}
    A_1(\bar\delta)=e^{F_1 T}+\begin{bsmallmatrix}
            0 & \bar\delta & 0\\0 & 0 & 0\\0& 0 &0
        \end{bsmallmatrix},~B_1(\bar\delta)=\int_0^{T} e^{F_1 s}g_1\textup{d}s  + \begin{bsmallmatrix}
            0\\4\bar\delta\\0
        \end{bsmallmatrix},
\end{equation*}
\begin{equation*}
    A_2(\bar\delta)=e^{F_2 T}+\begin{bsmallmatrix}
            0 & 0 & 0\\ 2\bar\delta & 0 & 0\\0&0&0
        \end{bsmallmatrix},~B_2(\bar\delta)=\int_0^{T} e^{F_2 s}g_2\textup{d}s  + \begin{bsmallmatrix}
            2.8\bar\delta\\0\\0
        \end{bsmallmatrix},
\end{equation*}
where
\begin{equation*}
    F_1\!=\!\begin{bsmallmatrix}
            -3 & -6 & 3\\2 & 2 & -3\\1.6& 0 &-2
        \end{bsmallmatrix},~F_2\!=\!\begin{bsmallmatrix}
            1 & 3 & 3\\-0.2 & -3 & -3\\0&0&-2
        \end{bsmallmatrix},~g_1\!=\!\begin{bsmallmatrix}
            0.5\\0\\0
        \end{bsmallmatrix},~g_2\!=\!\begin{bsmallmatrix}
            0\\0\\0.5
        \end{bsmallmatrix},
\end{equation*}
and parameter $\bar\delta$ is assumed to be unknown and time-varying, but bounded by $|\bar\delta|\<d$ with $d\>0$. Here, the nominal matrices are chose as $\bar A_i=A_i(0)$ and $\bar B_i=B_i(0),~\forall j\in\mathbb{K}$, and cycle $\nu=\{1,2\}$. With those nominal matrices, Assumption~\ref{ass:nominal_system} is verified, yielding
\begin{equation*}    P_1=\begin{bsmallmatrix}2.43&13.34&-9.17\\13.34&95.4&-57.69\\-9.17&-57.69&70.21\end{bsmallmatrix},\ P_2=\begin{bsmallmatrix}13.4&13.21&-9.44\\13.21&17.84&-8.34\\-9.44&-8.34&16.03\end{bsmallmatrix},
\end{equation*}
and the nominal limit cycle given by $ \rho_1=\begin{bsmallmatrix}3.84&
-0.65&
0.36\end{bsmallmatrix}^{\!\top}$, 
$\rho_2=\begin{bsmallmatrix}1.12&0.014&1.1042\end{bsmallmatrix}^{\!\top}$.

Considering the previous vectors $\{\rho_i\}_{i\in\mathbb D_\nu}$ and matrices $\{P_i\}_{i\in\mathbb D_\nu}$ satisfying Assumption~\ref{ass:nominal_system} and setting $\gamma=0.125$ and $\alpha=0.0001$, the optimization problem of Theorem~\ref{op:problem}  is verified using the CVX toolbox in MATLAB. The attractor $\mathcal{S}_{\nu}$ is estimated and compared to attractor 
\begin{equation*}
    \begin{split}
        \widetilde{\mathcal S}_\nu \!\!:=\!\! \left\{ \xi\!\in\!\mathbb H_\nu^\star, \mbox{ s.t. } \begin{bsmallmatrix}
    \chi_{1|1}\!-\!\rho_{\theta_1}\\ \chi_{0|0}\!-\!\chi_{1|0}
\end{bsmallmatrix}^{\!\top}\widetilde{\Omega}_{\theta_1}\begin{bsmallmatrix}
    \chi_{1|1}\!-\!\rho_{\theta_1}\\ \chi_{0|0}\!-\!\chi_{1|0}
\end{bsmallmatrix}\!\!\<\!1\!\right\},
    \end{split}
\end{equation*}
where $\widetilde{\Omega}_{\theta_1}= \begin{bsmallmatrix}
 P_{\theta_1}\!-\!R & R \\ 
 \star & Q+R
\end{bsmallmatrix}$, obtained in \cite{portilla2026a}. 

In Table~\ref{table:1}, the comparison is quantified by computing the ellipsoidal volume of each attractor by taking into account that, for any ellipsoid $\mathcal E(M_i)$ with $M_i\succ\0$, we have
\begin{equation}\label{eq:vol_ellip}
    \log\text{Vol}(\mathcal E(M_i))\propto c_i:=\log\det(M_i^{-1}),
\end{equation}
particularly, $M_i=\Omega_i$ for $\mathcal{S}_{\nu}$ and $M_i=\widetilde\Omega_i$ for $\widetilde{\mathcal{S}}_{\nu}$. 
\begin{table}[!t]
\begin{center}
\caption{Characterization of the ellipsoidal attractor volume using constant $c_i$ in \eqref{eq:vol_ellip} in log-scale.}\label{table:1}
\begin{tabular}{c|cc|cc|}
\cline{2-5}
                       & \multicolumn{2}{c|}{$d=0.0037$}    & \multicolumn{2}{c|}{$d=0.04$}    \\ \hline
\multicolumn{1}{|c|}{Attractor} & \multicolumn{1}{c|}{$c_1$} & \multicolumn{1}{c|}{$c_2$} & \multicolumn{1}{c|}{$c_1$} &  \multicolumn{1}{c|}{$c_2$}\\ \hline
\multicolumn{1}{|c|}{$\mathcal{S}_{\nu}$} & \multicolumn{1}{c|}{$-41.58$} & \multicolumn{1}{c|}{$-40.47$} & \multicolumn{1}{c|}{$-16.38$} & \multicolumn{1}{c|}{$-15.27$}  \\ \hline
\multicolumn{1}{|c|}{$\widetilde{\mathcal{S}}_{\nu}$} & \multicolumn{1}{c|}{$-29.89$} & \multicolumn{1}{c|}{$-28.79$} & \multicolumn{1}{c|}{\textup{---}} & \multicolumn{1}{c|}{\textup{---}} \\ \hline
\end{tabular}
\end{center}
\vspace{-0.65cm}
\end{table}
For the parameter $d=0.0037$, Table~\ref{table:1} shows that the volume of the ellipsoidal attractor for $\mathcal{S}_{\nu}$ with solution $\mu=1.91$ is smaller than that for $\widetilde{\mathcal{S}}_{\nu}$. This indicates that the attractor represented by $\mathcal{S}_{\nu}$ is less conservative. This is primarily due to the decision variable $\mu>0$ in $\mathcal{S}_{\nu}$, which allows for refinement. Furthermore, even when $d=0.04$, namely, a significant uncertainty in the system parameters, the verification of the optimization problem stated in Theorem~\ref{op:problem} is carried out with $\mu=0.0156$ as illustrated in Table~\ref{table:1}. In contrast, the attractor $\widetilde{\mathcal{S}}_{\nu}$ cannot be certified due to the unfeasible conservative conditions of \cite{portilla2026a}, which is indicated as $(-)$.

In \cite{portilla2026a}, the stability condition depends on a parameter $\bar \mu$ (not a decision variable), which was directly linked to condition \eqref{eq:LMI}, imposing a positivity condition on \eqref{eq:LMI} in such a way that $\bar A_{\nu(i)}^{\!\top} P_{\lfloor i+1\rfloor_\nu} \bar A_{\nu(i)}- (1-\bar \mu)P_i\prec0$. Clearly, the previous inequality was directly considered in Lemma 3 of \cite{portilla2026a}, which is restructured and refined in this contribution (compare to Lemma 3), as it was identified as a source of conservatism.  

Finally, for initial conditions $x_0=\begin{bsmallmatrix}-2&2&-2\end{bsmallmatrix}^{\!\top}$ and $\sigma_0=u_0,$ where $u_0$ is a random value such that $u_0\in\{1,2\}$, the system response of the closed-loop system \eqref{eq:model_xitheta}-\eqref{eq:control} is depicted in Figure~\ref{fig:predictor}. Indeed, this figure shows the effectiveness of the predictive control law, which evidences the robust convergence to the nominal limit cycle $\{\rho_i\}_{i\in\mathbb D_{\nu}}$, compensating for the input delay and the system uncertainties.
\begin{figure}[!t]   
     \centering
      \includegraphics[width=0.45\textwidth]{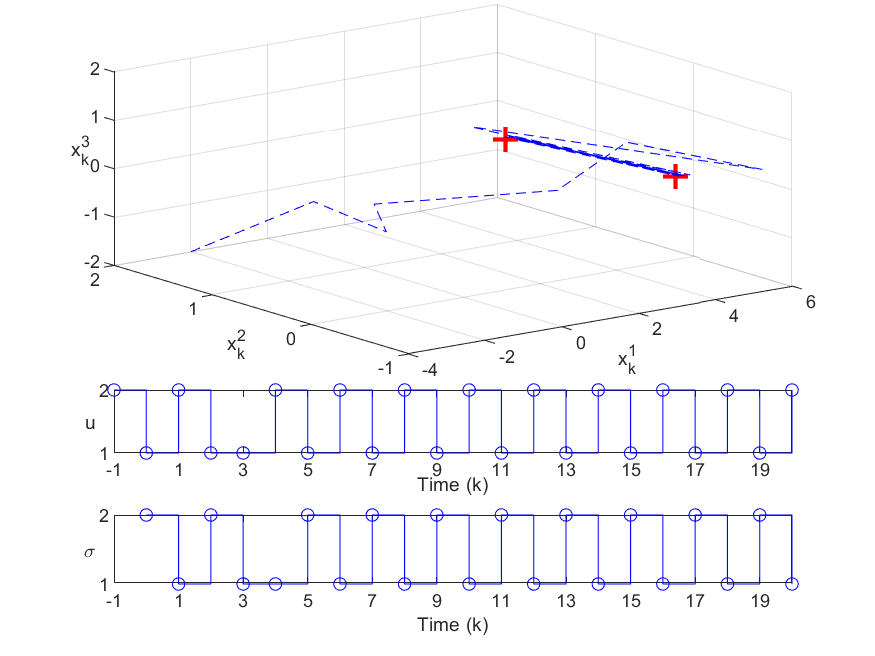}
      \vspace{-0.35cm}
    \caption{State trajectory of system \eqref{eq:model_xitheta} using the control law \eqref{eq:control}.}
      \label{fig:predictor}
    \vspace{-0.5cm}
\end{figure}

\section{Conclusions}\label{sec:conclusions}
This paper dealt with the robust stabilization of uncertain switched systems with a unitary input delay. Unlike in \cite{portilla2026a}, we provide a deeper study of the prediction scheme, yields to more tractable properties and an efficient structure for the Lyapunov function. These enhancements allowed us to obtain both less conservative stabilization conditions and a more accurate estimation of the attractor, which is illustrated through an academic example.
Finally,  the properties of the predictor and the structure of the Lyapunov function lay the groundwork for the case of a large delay, which will be addressed in future works.

\bibliographystyle{ieeetr}
\bibliography{bib_SAS}

\end{document}